\documentclass[conference]{IEEEtran}
\usepackage[utf8x]{inputenc}
\usepackage[cmex10]{amsmath}
\usepackage{amsthm}
\usepackage{amssymb}
\usepackage{url}
\usepackage{stmaryrd}
\usepackage[compatibility=false]{caption}
\usepackage{subcaption}
\usepackage{float}
\floatstyle{boxed}
\restylefloat{figure}
\usepackage{hyperref}

\DeclareMathOperator{\eval}{eval}
\DeclareMathOperator{\HSIstar}{\text{HSI*}}
\DeclareMathOperator{\HSI}{\text{HSI}}

\newtheorem{theorem}{Theorem}
\newtheorem{lemma}{Lemma}
\newtheorem{corollary}{Corollary}
\newtheorem{example}{Example}

\newcommand{\inter}[2]{\llbracket{#1}\rrbracket_{#2}}
\newcommand{\typeof}[1]{\llbracket{#1}\rrbracket}
\newcommand{\Nplus}{\mathbb{N}^+}

\newcommand{\inl}[1]{\iota_1{#1}}
\newcommand{\inr}[1]{\iota_2{#1}}
\newcommand{\case}[5]{\delta[{#1}|{#2}.{#3}|{#4}.{#5}]}

\newcommand{\casetwo}[5]{\delta[{#1}\\ & |{#2}.{#3}\\ & |{#4}.{#5}]}
\newcommand{\fst}[1]{\pi_1{#1}}
\newcommand{\snd}[1]{\pi_2{#1}}
\newcommand{\pair}[2]{\langle{#1},{#2}\rangle}

\newcommand{\Ecal}{\mathcal{E}}
\newcommand{\Estar}{\mathcal{E^*}}

\begin{document}
\title{Axioms and Decidability for Type Isomorphism\\ in the Presence of Sums}

\author{\IEEEauthorblockN{Danko Ilik}
\IEEEauthorblockA{Inria Saclay -- Île de France\\
1 rue Honor\'e d'Estienne d'Orves\\
Campus de l'\'Ecole Polytechnique\\
91120 Palaiseau, France\\
Email: danko.ilik@inria.fr}}

\maketitle

\begin{abstract}
We consider the problem of characterizing isomorphisms of types, or, equivalently, constructive cardinality of sets, in the simultaneous presence of disjoint unions, Cartesian products, and exponentials. Mostly relying on results about polynomials with exponentiation that have not been used in our context, we derive: that the usual finite axiomatization known as High-School Identities (HSI) is complete for a significant subclass of types; that it is decidable for that subclass when two types are isomorphic; that, for the whole of the set of types, a recursive extension of the axioms of HSI exists that is complete; and that, for the whole of the set of types, the question as to whether two types are isomorphic is decidable when base types are to be interpreted as finite sets. We also point out certain related open problems.
\end{abstract}


\section{Introduction}

The class of types built from Cartesian products ($\tau\times\sigma$), disjoint unions ($\tau+\sigma$) and function spaces ($\tau\to\sigma$) lies at the core of type systems for programming languages and constructive systems of Logic. How useful could a programming language be if it did not include pairing, enumeration and functions? Likewise, how useful would a constructive logic be, if it did not have conjunction, disjunction, and implication?

It is then a strange state of affairs that we still have open questions regarding general properties of this basic class of types. In this article, we revisit the problem of when two types can be considered to be isomorphic and we show that this problem can be tackled using existing results from Mathematical Logic.

Let us be more precise. The language of \emph{polynomials with exponentiation and with positive coefficients} is defined inductively by
\[
\Ecal \ni f,g ::= 1 ~|~ x_i ~|~ f+g ~|~ fg ~|~ g^f,
\]
where $x_i$ is a variable for $i\in\mathbb{N}$. This language determines the class of types that we are interested in, by the simple translation,
\begin{align*}
  \typeof{1} &= \mathbf{1}\\
  \typeof{x_i} &= \mathbf{x_i}\\
  \typeof{g^f} &= \typeof{f} \to \typeof{g}\\
  \typeof{f g} &= \typeof{f} \times \typeof{g}\\
  \typeof{f+g} &= \typeof{f} + \typeof{g},
\end{align*}
where $\mathbf{1}$ denotes the singleton type and $\mathbf{x_i}$ denotes a type variable (that can be instantiated during an interpretation in a concrete setting). By abuse of language, we will say that a type $\tau$ belongs to $\Ecal$ ($\tau\in\Ecal$) when a polynomial with exponentiation $f\in\Ecal$ exists such that $\tau=\typeof{f}$. Throughout the paper, we use the plain equality symbol ``$=$'' to stand for identity i.e. definitional equality.

The types of $\Ecal$ are inhabited by terms of a lambda calculus, in the usual way, following the typing system shown in Figure~\ref{fig:terms:typing}. The equality between two typed terms is the relation $=_{\beta\eta}$ given by the usual axioms in Figure~\ref{fig:terms:equality}.

Two types $\tau$ and $\sigma$ are called \emph{isomorphic} (notation $\tau\cong\sigma$) when there is a pair of lambda terms $\phi^{\tau\to\sigma}$ and $\psi^{\sigma\to\tau}$ that are mutually inverse, that is, $\lambda x. \phi (\psi x) =_{\beta\eta} \lambda x. x$ and $\lambda y. \psi (\phi y) =_{\beta\eta} \lambda y. y$. The importance of this notion in ``practice'' is as follows. 

In typed programming languages, to be able to say when two types are isomorphic amounts to being able to say when two programs implement essentially the same type signature: a program of type $\tau$ can be coerced back and forth to type $\sigma$ without loss of information. One can use this, for example, to search over a library of routines for a routine of a type coercible to the type needed by the programmer \cite{rittri91,cosmo05b}.

In propositional logic, knowing that two types are isomorphic means that the corresponding formulas (built from $\wedge, \vee$ and $\to$) are intuitionistically equivalent.

In Constructive Mathematics, type isomorphism coincides with the notion of \emph{constructive cardinality} \cite{mines88,lombardi11} that says that two sets (i.e., types) are isomorphic if they have indistinguishable structure, which is stronger than the classical notion of cardinality relying on ``number of elements''.

What is known about the isomorphism of types of $\Ecal$, in general, are the following facts, both proved in \cite{fiore06}.
\begin{theorem}[Soundness of HSI]\label{fac:hsi:sound}
  If $\HSI \vdash f\doteq g$, then $\typeof{f}\cong\typeof{g}$. In fact, we have:
  \[
  \HSI\vdash f\doteq g ~\Rightarrow~ \typeof{f}\cong\typeof{g} ~\Rightarrow~ \Nplus\vDash f\equiv g.
  \]
\end{theorem}
\begin{theorem}\label{fac:hsi:incomplete}
  Isomorphism is not finitely axiomatizable, that is, for no finite set of axioms T can we show that $\typeof{f}\cong\typeof{g}$ always implies $\text{ T } \vdash f\doteq g$.
\end{theorem}
The notation $\text{ T } \vdash f\doteq g$ means that there is a formal derivation of the equation $f \doteq g$ in the derivation system shown in part~(b) of Figure~\ref{fig:hsi}, from the axioms of the set T; therefore, $\HSI \vdash f\doteq g$ means that the equation is derivable from the finite set of axioms shown in part~(a) of Figure~\ref{fig:hsi} (HSI stands for ``High-School Identities'', see Subsection~\ref{subsec:tarski} below).

Having only Theorem~\ref{fac:hsi:sound} and Theorem~\ref{fac:hsi:incomplete} is surprisingly little if we compare to what is known for the fragments of $\Ecal$ that do not mix $g^f$ and $g+f$ \emph{simultaneously}. For those fragments, we have, as shown in \cite{soloviev81,soloviev83,fiore06}, soundness \emph{and completeness} with respect to the suitable restriction of HSI, we have moreover equivalence with truth in the standard model of positive natural numbers $\Nplus$ (see Subsection~\ref{subsec:tarski} for definition of truth in $\Nplus$),
\[
\typeof{f}\cong\typeof{g} ~\Leftrightarrow~ \Nplus\vDash f\equiv g,
\]
and, consequently, the \emph{decidability} of $\tau\cong\sigma$ for any $\tau$ and $\sigma$ of those fragments. 

In this paper, we will address both the questions of completeness and decidability for $\Ecal$, in simultaneous presence of $g^f$ and $g+f$. In Section~\ref{sec:subclasses}, we will bring up the relevance of certain subclasses of types going back to Levitz, and explain how results of Henson, Rubel, Gurevič, Richardson, and Macintyre, allow to show that type isomorphism for those subclasses are complete with respect to HSI, and decidable. In Section~\ref{sec:hsistar}, using Wilkie's positive solution of Tarski's High-School Algebra Problem (see next subsection), we will establish the same properties for the whole of $\Ecal$ (decidability is proved for base types interpreted as finite sets.). In Section~\ref{sec:conclusion}, we will mention related open problems, in particular about having efficient means of deciding the type isomorphisms.

\subsection{Tarski's High-School Algebra Problem}\label{subsec:tarski}

The questions that we are interested in are related to questions regarding polynomials with exponentiation from the class $\Ecal$, posed by Skolem \cite{skolem56} and Tarski \cite{martin73} in the 1960's. Especially relevant is the question known as Tarski's High-School Algebra Problem\footnote{For various results around this problem, please look at the survey articles \cite{burris93} and \cite{burris04}.}: can all equations that are true in the standard model of positive natural numbers $(\Nplus\vDash f\equiv g)$ be derived inside the derivation system of HSI from parts (a) and (b) of Figure~\ref{fig:hsi} $(\HSI\vdash f\doteq g)$? This is a \emph{completeness} question.

The meaning of $\Nplus\vDash f\equiv g$ is the standard model theoretic one: for any replacement of the variables of $f$ and $g$ by elements of $\Nplus$, one computes the same positive natural number. The converse (\emph{soundness}),
\[
\HSI\vdash f\doteq g ~\Rightarrow~ \Nplus\vDash f\equiv g,
\]
can easily be proved.

\begin{figure*}
    \begin{subfigure}[b]{0.4\textwidth}
        \centering
        \begin{align*}
          f &\doteq f\\
          f+g &\doteq g+f \\
          (f+g)+h &\doteq f+(g+h)\\
          f g &\doteq g f\\
          (f g) h &\doteq f (g h)\\
          f(g+h) &\doteq f g + f h\\
          1 f &\doteq f\\
          f^1&\doteq f\\
          1^f&\doteq 1\\
          f^{g+h}&\doteq f^g f^h\\
          (f g)^h&\doteq f^h g^h\\
          (f^g)^h&\doteq f^{g h}\\
        \end{align*}
        \caption{Axioms of HSI}\label{fig:hsi:a}
    \end{subfigure}
    ~
    \begin{subfigure}[b]{0.5\textwidth}
      \begin{subfigure}[b]{1.0\textwidth}
        \centering
        \begin{align*}
          \frac{f\doteq g}{g\doteq f}&~&
          \frac{f\doteq g\quad g\doteq h}{f\doteq h}&~&
          \frac{f_1\doteq g_1\quad f_2\doteq g_2}{f_1^{f_2}\doteq g_1^{g_2}}
        \end{align*}
        \begin{align*}
          \frac{f_1\doteq g_1\quad f_2\doteq g_2}{f_1+f_2\doteq g_1+g_2}&~&
          \frac{f_1\doteq g_1\quad f_2\doteq g_2}{f_1 f_2\doteq g_1 g_2}
        \end{align*}
        \caption{Equality and congruence rules}\label{fig:hsi:b}
      \end{subfigure}

      \begin{subfigure}[b]{1.0\textwidth}
        \centering
        \begin{align*}
          t_{1} &\doteq 1\\
          t_{x_i} &\doteq x_i\\
          t_{z u} &\doteq t_z t_u\\
          t_{z+u} &\doteq t_z + t_u\\
          t_z &\doteq t_{u} &(\text{when } \mathbb{N}^+\vDash z\equiv u)
        \end{align*}
        \caption{Additional axioms of HSI*}\label{fig:hsi:c}
      \end{subfigure}
    \end{subfigure}
    \begin{center}
          Notation: $f, g, h$ denote polynomials with exponentiation (possibly containing $t_z$ terms in the congruence and equality rules), while $z, u$ denote ordinary polynomials (i.e., without non-constant exponents) with possibly  negative monomial coefficients.
    \end{center}
    \caption{The derivation system of (Extended) High-School Identities}\label{fig:hsi}
\end{figure*}

Martin~\cite{martin73} was the first to show that, if we exclude the axioms mentioning the constant $1$ from HSI, the derivation system is incomplete, since it can not derive the equality
\[
(x^z+x^z)^w(y^w+y^w)^z = (x^w+x^w)^z(y^z+y^z)^w.
\]
Wilkie~\cite{wilkie00} generalized Martin's equality to the whole of HSI, giving the equation
\begin{equation}
  \label{eq:wilkie}
  (A^x+B^x)^y(C^y+D^y)^x = (A^y+B^y)^x(C^x+D^x)^y,
\end{equation}
where $A=1+x, B=1+x+x^2, C=1+x^3, D=1+x^2+x^4$.
He showed (\ref{eq:wilkie}) to be non-derivable in HSI, even though it is true in $\Nplus$. We thus have
\[
\forall f,g\in\Ecal (\HSI\vdash f\doteq g ~\Rightarrow~ \Nplus\vDash f\equiv g),
\]
but
\[
\forall f,g\in\Ecal (\Nplus\vDash f\equiv g ~\not\Rightarrow~ \HSI\vdash f\doteq g),
\]
which constitutes a \emph{negative} solution to Tarski's original question.

Gurevič \cite{gurevic90} further showed that one can not ``repair'' HSI by extending it with \emph{any} finite list of axioms. He generalized Wilkie's (\ref{eq:wilkie}) to the infinite sequence of equations
\begin{equation}\tag{$G_n$}
  \label{eq:gurevic}
  (A^{2^x}+B_n^{2^x})^x(C_n^x+D_n^x)^{2^x} = (A^x+B_n^x)^{2^x}(C_n^{2^x}+D_n^{2^x})^x,
\end{equation}
where $A=x+1, B_n=1+x+x^2+\cdots+x^{n-1}, C_n=1+x^n, D_n=1+x^2+x^4+\cdots+x^{2(n-1)}$, and showed that for any finite extension T of HSI there is an odd $n>3$ such that T can not prove the equality ($G_n$) although $\Nplus\vDash G_n$.

Fiore, Di Cosmo and Balat \cite{fiore06} showed that Gurevič's equations can be interpreted as type isomorphisms, establishing the mentioned Theorem~\ref{fac:hsi:incomplete}.

\subsection{Decidability of Arithmetic Equality for $\Ecal$}

A separate question of more general interest is that of the decidability of equality between polynomials with exponentiation, that is, whether there is a procedure for deciding if $\Nplus\vDash f\equiv g$ holds or not, for any $f,g\in\Ecal$. It was first addressed by Richardson \cite{richardson69}, who proved decidability for the univariate case (expressions of $\Ecal$ in one variable). Later, Macintyre \cite{macintyre81} showed the decidability for the multivariate case i.e. for the whole of $\Ecal$.

\begin{theorem}\label{fac:E:decidable}
There is a recursive procedure that decides, for any $f,g\in\Ecal$, whether $\Nplus\vDash f\equiv g$ holds or not.
\end{theorem}

However, we cannot use the decidability result for $\Nplus$ to conclude decidability of type isomorphisms for $\Ecal$, because, although we do have that (by Theorem~\ref{fac:hsi:sound})
\[
  \HSI\vdash f\doteq g ~\Rightarrow~ \typeof{f}\cong\typeof{g} ~\Rightarrow~ \Nplus\vDash f\equiv g,
\]
a proof of
\[
  \typeof{f}\cong\typeof{g} ~\Leftarrow~ \Nplus\vDash f\equiv g
\]
is not known, and HSI is not complete:
\[
  \HSI\vdash f\doteq g~\not\Leftarrow~ \Nplus\vDash f\equiv g.
\]

\section{Subclasses of $\Ecal$ Complete for HSI}\label{sec:subclasses}

One of the things that has not been exploited in the literature on type isomorphism is the line of research on \emph{subclasses} of polynomials with exponentiation for which the axioms of HSI \emph{are} complete.

In \cite{levitz75}, while studying the relation of eventual dominance for polynomials with exponentiation, Levitz isolated the class of expressions in one variable, built by the inductive definition
\[
\mathcal{S} \ni f,g ::= 1 ~|~ x ~|~ f+g ~|~ fg ~|~ x^f ~|~ n^f,
\]
where $n$ is a numeral. Henson and Rubel \cite{henson84} extended it to the multivariate class defined by
\[
\mathcal{L}(S) \ni f,g ::= s ~|~ x_i ~|~ f+g ~|~ fg ~|~ {s'}^f ~|~ x_i^f ~|~ (x_i^{s'})^f,
\]
where $S$ is an arbitrary set of positive real constants, $s,s'\in S$, $s'>1$, and they proved all true equalities between expressions from $\mathcal{L}(S)$ to be derivable from HSI. They also conjectured that the result could be extended to the class defined by
\[
\mathcal{R}(S) \ni f,g ::= s ~|~ x_i ~|~ f+g ~|~ fg ~|~ p^f,
\]
where $p$ is an ordinary polynomial with coefficients in $S$, and they remarked that Wilkie's counterexample lies ``just outside'' the class $\mathcal{R}(S)$.

Finally, Gurevič \cite{gurevic93} showed that HSI are complete for the proper extension $\mathcal{L}$ of $\mathcal{R}(S)$ defined by\footnote{We stick to the original notations $\mathcal{L}(S)$ and $\mathcal{L}$, although the latter also depends on the set $S$ and we have $\mathcal{L}(S)\subsetneq \mathcal{R}(S)\subsetneq \mathcal{L}$.}
\[
\mathcal{L} \ni f,g ::= s ~|~ x_i ~|~ f+g ~|~ f g ~|~ l^f,
\]
where $l\in\Lambda$ is defined by
\[
\Lambda \ni f,g ::= s ~|~ x_i ~|~ f+g ~|~ f g ~|~ l_0^f,
\]
and $l_0\in \Lambda$ has no variables.

\begin{theorem}\label{fac:L:complete}
  For all $f,g\in\mathcal{L}$,
  \[\Nplus \vDash f\equiv g ~\Rightarrow~ \HSI\vdash f\doteq g.\]
\end{theorem}

For our purposes, it suffices to take $S=\{1\}$, and we will henceforth use $\mathcal{L}$ specialized to this $S$.

\begin{example}\label{example2}
Wilkie's equation~(\ref{eq:wilkie}) deals with terms that do not belong to the class $\mathcal{L}$. Although $A, B, C, D \in \Lambda\subset\mathcal{L}$, and hence $A^x, B^x, C^x, D^x, A^x+B^x, C^x+D^x \in \mathcal{L}\setminus\Lambda$, we have $(A^x+B^x)^y, (C^x+D^x)^y \notin\mathcal{L}$, because bases of exponentiation are not allowed to contain bases of exponentiation that contain variables.
\end{example}

\begin{example} The term\footnote{These terms correspond to simply typed versions of an induction axiom for decidable predicates, in ``curried'' and ``uncurried'' variant.}
  \begin{equation}
    \label{eq:uncurried}
    (y+z)^{x(y+z)(y+z)^{x(y+z)}}\in\mathcal{L},
  \end{equation}
but
  \begin{equation}
    \label{eq:curried}
    \left({\left(\left(y+z\right)^x\right)}^{((y+z)^{y+z})^x}\right)^{y+z}\notin\mathcal{L},
  \end{equation}
although the two terms are inter-derivable using the HSI axioms. This means that, even though HSI is complete for $\mathcal{L}$, there is room for extension of $\mathcal{L}$ to subclasses that are still finitely axiomatizable by HSI. In other words, $\mathcal{L}$ is not closed under HSI-derivability.
\end{example}

Theorems \ref{fac:hsi:sound}, \ref{fac:E:decidable} and \ref{fac:L:complete} allow us to conclude the following.
\begin{corollary}[Completeness for $\mathcal{L}$] For all $f,g\in\mathcal{L}$, if $\typeof{f}\cong\typeof{g}$ then $\HSI\vdash f\doteq g$.
\end{corollary}

\begin{corollary}[Decidability for $\mathcal{L}$] There is an algorithm that decides, for all $f,g\in\mathcal{L}$,  whether $\typeof{f}\cong\typeof{g}$ or not.
\end{corollary}

\section{The Extended High-School Identities}\label{sec:hsistar}

As explained in Subsection~\ref{subsec:tarski}, Wilkie's negative solution of Tarski's problem, together with the generalization of Gurevič, was used by Fiore, Di Cosmo, and Balat, to show the incompleteness of HSI, and the impossibility of a finite axiomatization, for type isomorphism over $\Ecal$.

However, in the paper~\cite{wilkie00}, a \emph{positive} solution to Tarski's problem was also given. Namely, Wilkie showed that HSI is almost complete: by extending it with all equations that hold between ordinary positive polynomials\footnote{A polynomial is positive if it computes to a positive natural number for every replacement of variables by positive natural numbers.} --- that is, positive polynomials without exponents containing variables, but possibly with negative monomial coefficients --- one obtains an axiomatization (HSI* from Figure~\ref{fig:hsi}) which is complete for all true equations in $\mathbb{N}^+$ between expressions of $\Ecal$. Since equality of ordinary (positive) polynomials is decidable, we have a recursive procedure for determining if an equation belongs to the set of axioms, and therefore equality between terms of $\Ecal$ --- although not finitely axiomatizable --- is recursively axiomatizable.

To be more precise, let $\Estar$ be the language extending $\Ecal$ with a constant $t_z$ for every ordinary positive polynomial $z$:
\[
\Estar \ni f,g ::= t_z ~|~ 1 ~|~ x_i ~|~ g^f ~|~ f g ~|~ f+g.
\]
The system of axioms of HSI* is the extension of the system HSI --- that applies only to expressions of the original language $\Ecal$ --- with the axioms given in part~(c) of Figure~\ref{fig:hsi} --- that apply to a strict subset of the extended language $\Estar$. The two types of axioms can ``interact'' through the derivation system of part~(b) of Figure~\ref{fig:hsi}, thus making possible equalities between expressions of the full language $\Estar$.

The left-hand side of the new axioms is always an expression of form $t_z$, while the right hand side can either be an expression $t_u$ --- whenever $z$ and $u$ are equal polynomials (which can be decided by bringing them into canonical form) --- or an expression reflecting the structure of $z$ when possible: when $z$ has no negative coefficients (we use in that case letters $p,q$ instead of $z,u$), we can fully reflect it into the language, that is, one can prove $\HSIstar\vdash t_p \doteq p$. This representation of the axioms is inspired from the one of Asatryan \cite{asatryan08}. Note that, for any $t_z$, there are infinitely many axioms having $t_z$ on the left hand side that can be used.

We can now state Wilkie's result.

\begin{theorem}[Completeness of HSI*]\label{thm:wilkie} For all $f,g\in\Ecal$ (that is, all $f,g$ of $\Estar$ that do \textbf{not} contain $t_z$-symbols), we have that $\mathbb{N}^+ \vDash f\equiv g$ implies $\HSIstar \vdash f\doteq g$.
\end{theorem}

This is a statement concerning terms of $\Ecal$ (the original language), in the proof of which terms of $\Estar$ (the extended language) are used. In this respect, it is reminiscent of meta-mathematical statements like Hilbert's $\epsilon$-elimination theorems \cite{hilbert01} or Henkin's version of Gödel's completeness theorem \cite{henkin49}. 

Using theorems~\ref{fac:hsi:sound} and~\ref{thm:wilkie}, we immediately obtain the completeness of $\HSIstar$ for type isomorphism over $\Ecal$.

\begin{corollary}\label{cor:hsistar:complete} Given $f, g\in\Ecal$ such that $\typeof{f}\cong\typeof{g}$, we have that $\HSIstar\vdash f\doteq g$.
\end{corollary}

We now move on to the decidability question. We will show that derivations of HSI* can be interpreted as type isomorphisms, which suffices, since then the circuit 
\[
\typeof{f}\cong\typeof{g} ~\Rightarrow~ \Nplus\vDash f\equiv g~\Rightarrow~ \HSIstar\vdash f\doteq g ~\Rightarrow~ \typeof{f}\cong\typeof{g}
\]
allows one to use Macintyre's decidability results for $\Nplus$ (Theorem~\ref{fac:E:decidable}) to conclude decidability of type isomorphism over $\Ecal$.

At first thought, interpreting the new $t_z$ symbols might seem problematic, since negative monomial coefficients in $z$ would imply the use of some kind of negative types. However, we also have the additional property that $z$ is positive, which means that if we instantiate its variables with positive natural numbers (positive types), we will obtain a positive natural number (positive type). 

In this paper, we will work with the restriction that base types are finite sets. Although the method does work for base types isomorphic to ordinals in Cantor normal form\footnote{Subtraction $\alpha-\beta$ between two such ordinals can be defined when we know that $\alpha<\beta$. Since we can always rewrite $t_z$ as $t_p-t_q$, and we know $p<q$, an ordinal in Cantor normal form can always be computed for $t_z$ whenever ones for $p$ and $q$ are given.}, that requires a careful constructive treatment of ordinals beyond the scope of this paper.

For HSI, one can keep the interpretation of base types implicit: soundness of HSI equations as type isomorphisms is proved uniformly, regardless of the actual interpretations of base types. For HSI*, we will need to be explicit about interpretation, that is, we will prove the soundness theorem \emph{point-wise}. We will thus introduce an explicit environment $\rho$ mapping variables to types and extend the interpretation $\typeof{\cdot}$ for the extra $t_z$-terms.
\begin{align*}
  \inter{1}{\rho} &= \mathbf{1}\\
  \inter{x_i}{\rho} &= \rho(x_i)\\
  \inter{g^f}{\rho} &= \inter{f}{\rho} \to \inter{g}{\rho}\\
  \inter{f g}{\rho} &= \inter{f}{\rho} \times \inter{g}{\rho}\\
  \inter{f+g}{\rho} &= \inter{f}{\rho} + \inter{g}{\rho}\\
  \inter{t_z}{\rho} &= \underbrace{1+1+\cdots+1}_{k\text{-times}} = \mathbf{k} \quad \text{ where } k=\eval(t_z,\rho)
\end{align*}
The number $\eval(t_z,\rho)$ is the result of evaluating $z$ for the variables interpreted in $\rho$ by positive natural numbers. We also denote the type 
\[
\underbrace{1+1+\cdots+1}_{k\text{-times}}
\]
with bold-face $\mathbf{k}$.

\begin{figure*}
    \begin{subfigure}[b]{1.0\textwidth}
        \centering
        \[
        \Lambda^+ \ni M,N,P ::= x ~|~ \star ~|~ \lambda x. M ~|~ M N ~|~ \inl{M}  ~|~ \inr{M} ~|~ \case{M}{x}{N}{x}{P} ~|~ \pair{M}{N} ~|~ \fst{M} ~|~ \snd{M}
        \]
        \caption{Raw language of lambda terms}
    \end{subfigure}

    \begin{subfigure}[b]{1.0\textwidth}
        \centering
        \begin{align*}
          \frac{x^\tau \in \Gamma}{\Gamma \Vdash x^\tau} & & \frac{}{\Gamma \Vdash \star^{\mathbf{1}}} & & \frac{\Gamma \cup x^\tau \Vdash M^\sigma}{\Gamma\Vdash (\lambda x. M)^{\tau\to\sigma}} & & \frac{\Gamma\Vdash M^{\tau\to\sigma}\quad \Gamma\Vdash N^\tau}{\Gamma\Vdash (M N)^\sigma}
        \end{align*}
        \begin{align*}
          \frac{\Gamma\Vdash M^\tau}{\Gamma\Vdash(\inl{M})^{\tau + \sigma}} & & \frac{\Gamma\Vdash M^\sigma}{\Gamma\Vdash (\inr{M})^{\tau + \sigma}} & & \frac{\Gamma\Vdash M^{\tau+\sigma}\quad \Gamma\cup x^\tau\Vdash N^\rho \quad \Gamma\cup x^\sigma\Vdash P^\rho}{\Gamma \Vdash \case{M}{x}{N}{x}{P}^\rho}
        \end{align*}
        \begin{align*}
          \frac{\Gamma\Vdash M^\sigma\quad \Gamma\Vdash N^\tau}{\Gamma\Vdash \pair{M}{N}^{\sigma\times\tau}} & & \frac{\Gamma\Vdash M^{\sigma\times\tau}}{\Gamma\Vdash(\fst{M})^\sigma} & & \frac{\Gamma\Vdash M^{\sigma\times\tau}}{\Gamma\Vdash(\snd{M})^\tau}
        \end{align*}
        \caption{Typing system (well-formed lambda terms)}
    \end{subfigure}
    \caption{Inhabitation of types of $\Ecal$ with lambda terms}\label{fig:terms:typing}
\end{figure*}

\begin{figure*}
    \begin{subfigure}[b]{1.0\textwidth}
        \centering
        \begin{align*}
           M &=_{\beta\eta} \star & \text{ for any term } M \text{ of type } \mathbf{1}\\
           (\lambda x. M) N &=_{\beta\eta} M\{N/x\} & \text{ where $x$ is not a variable of $N$ }\\
           M &=_{\beta\eta} \lambda x. M x & \text{ where $x$ is not a variable of $M$ } \\
           \case{\inl{M}}{x}{N}{y}{P} &=_{\beta\eta} N\{M/x\} & \text{ where $x$ is not a variable of $M$ } \\
           \case{\inr{M}}{x}{N}{y}{P} &=_{\beta\eta} P\{M/y\} & \text{ where $y$ is not a variable of $M$ } \\
           N\{M/z\}&=_{\beta\eta} \case{M}{x}{N\{\inl{x}/z\}}{y}{N\{\inr{y}/z\}} & \text{ where $z$ is not a variable of $M$ } \\
           \fst{\pair{M}{N}} &=_{\beta\eta} M\\
           \snd{\pair{M}{N}} &=_{\beta\eta} N\\
           M &=_{\beta\eta} \pair{\fst{M}}{\snd{M}}\\
        \end{align*}
    \end{subfigure}
    \begin{center}
      The full relation $=_{\beta\eta}$ is the reflexive, symmetric, transitive, and congruent closure of the above.
    \end{center}
    \caption{Beta-eta equality between lambda terms of the same type}\label{fig:terms:equality}
\end{figure*}

\begin{theorem}\label{thm:hsistar:sound} Let $f, g\in \Estar$.
  If $\HSIstar \vdash f\doteq g$ then $\inter{f}{\rho}\cong\inter{g}{\rho}$ for any $\rho$ that interprets variables by types of form $\mathbf{k}$.
\end{theorem}
\begin{proof} The proof is by induction on the derivation.
  We first give explicit isomorphisms for the axioms of HSI:
  \begin{itemize}
  \item $f  \doteq f$ is interpreted with the identity lambda term $\lambda x.x$ in both directions;
  \item $f+g  \doteq g+f $ is interpreted by $\lambda x. \case{x}{x_1}{\inr{x_1}}{x_2}{\inr{x_2}}$ in both directions;
  \item $(f+g)+h  \doteq f+(g+h)$ is interpreted by 
    \[\case{x}{x_1}{\inl{\inl{x}}}{x_2}{\case{x_2}{x_{21}}{\inl{\inr{x_2}}}{x_{22}}{\inr{x_{22}}}}\]
    and
    \[\case{x}{x_1}{\case{x_1}{x_{11}}{\inl{x_{11}}}{x_{12}}{\inr{\inl{x_{12}}}}}{x_2}{\inr{\inr{x_2}}};\]
  \item $f g  \doteq g f$ is interpreted by $\lambda x.\pair{\snd{x}}{\fst{x}}$ in both directions;
  \item $(f g) h  \doteq f (g h)$ is interpreted by \[\lambda x. \pair{\pair{\fst{x}}{\fst{\snd{x}}}}{\snd{\snd{x}}}\] and \[\lambda x. \pair{\fst{\fst{x}}}{\pair{\snd{\fst{x}}}{\snd{x}}};\]
  \item $f(g+h)  \doteq f g + f h$ is interpreted by \[\lambda x. \case{\snd{x}}{x_1}{\inl{\pair{\fst{x}}{x_1}}}{x_2}{\inr{\pair{\fst{x_1}}{\inr{\snd{x_1}}}}} \] and \[\lambda x. \case{x}{x_1}{\pair{\fst{x_1}}{\inl{\snd{x_1}}}}{x_2}{\pair{\fst{x_2}}{\inr{\snd{x_2}}}};\]
  \item $f 1 \doteq f$ is interpreted by $\lambda x. \fst{x}$ and $\lambda x. \pair{x}{\star}$;
  \item $f^1 \doteq f$ is interpreted by $\lambda x. x \star$ and $\lambda x y. x$;
  \item $1^f \doteq 1$ is interpreted by $\lambda x. \star$ and $\lambda x y. \star$;
  \item $f^{g+h} \doteq f^g f^h$ is interpreted by
    \[
    \lambda x. \pair{\lambda y. x (\inl{y})}{\lambda y. x (\inr{y})}
    \]
    and
    \[
    \lambda x y. \case{y}{y_1}{(\fst{x})y_1}{y_2}{(\snd{x})y_2};
    \]
  \item $(f g)^h \doteq f^h g^h$ is interpreted by $\lambda x. \pair{\lambda y. \fst{x y}}{\lambda y. \snd{x y}}$ and $\lambda x y. \pair{(\fst{x})y}{(\snd{x})y}$;
  \item $(f^g)^h \doteq f^{g h}$ is interpreted by $\lambda x y. x (\fst{y}) (\snd{y})$ and $\lambda x y z. x \pair{z}{y}$.
  \end{itemize}

  The congruence and equality rules are handled using the induction hypotheses:
  \begin{itemize}
  \item Given an interpretation of $f\doteq g$, i.e. $\Phi:\inter{f}{\rho}\to\inter{g}{\rho}$ and $\Psi:\inter{g}{\rho}\to\inter{f}{\rho}$ such that
    \begin{align*}
      \lambda x. \Phi(\Psi x) &=_{\beta\eta} \lambda x. x\\
      \lambda y. \Psi(\Phi y) &=_{\beta\eta} \lambda y. y,
    \end{align*}
    we just swap the order of the two equations in order to interpret $g\doteq f$.
  \item Given interpretations of $f\doteq g$ and $g\doteq h$ by four terms
    \begin{align*}
      \Phi_1 &:\inter{f}{\rho}\to\inter{g}{\rho} & \Phi_2 &:\inter{g}{\rho}\to\inter{h}{\rho}\\
      \Psi_1 &:\inter{g}{\rho}\to\inter{f}{\rho} & \Psi_2 &:\inter{h}{\rho}\to\inter{g}{\rho},
    \end{align*}
    we interpret $f\doteq h$ by composing $\Phi_1$ and $\Phi_2$, and $\Psi_1$ and $\Psi_2$.
  \item Given interpretations of $f_1\doteq g_1$ and $f_2\doteq g_2$ by four terms
    \begin{align*}
      \Phi_1 &:\inter{f_1}{\rho}\to\inter{g_1}{\rho} & \Phi_2 &:\inter{f_2}{\rho}\to\inter{g_2}{\rho}\\
      \Psi_1 &:\inter{g_1}{\rho}\to\inter{f_1}{\rho} & \Psi_2 &:\inter{g_2}{\rho}\to\inter{f_2}{\rho},
    \end{align*}
    we interpret $f_1^{f_2}\doteq g_1^{g_2}$ by using the terms
    \begin{align*}
      \Phi &= \lambda x. \lambda y. \Phi_1(x (\Psi_2 y))\\
      \Psi &= \lambda x. \lambda y. \Psi_1(x (\Phi_2 y)).
    \end{align*}
    The fact that $\Phi$ and $\Psi$ are mutually inverse w.r.t. $=_{\beta\eta}$ is proved by using the $\eta$-axiom
    \[
    \lambda x. \lambda y. x y =_{\beta\eta} \lambda x. x.
    \]
  \item Given interpretations of $f_1\doteq g_1$ and $f_2\doteq g_2$ by four terms
    \begin{align*}
      \Phi_1 &:\inter{f_1}{\rho}\to\inter{g_1}{\rho} & \Phi_2 &:\inter{f_2}{\rho}\to\inter{g_2}{\rho}\\
      \Psi_1 &:\inter{g_1}{\rho}\to\inter{f_1}{\rho} & \Psi_2 &:\inter{g_2}{\rho}\to\inter{f_2}{\rho},
    \end{align*}
    we interpret $f_1 f_2\doteq g_1 g_2$ by using the terms
    \begin{align*}
      \Phi &= \lambda x. \pair{\Phi_1(\fst{x})}{\Phi_2(\snd{x})}\\
      \Psi &= \lambda y. \pair{\Psi_1(\fst{y})}{\Psi_2(\snd{y})}.
    \end{align*}
    The fact that $\Phi$ and $\Psi$ are mutually inverse w.r.t. $=_{\beta\eta}$ is proved by using the $\eta$-axiom
    \[
    \lambda y. \pair{\fst{y}}{\snd{y}} =_{\beta\eta} \lambda y. y.
    \]
  \item Given interpretations of $f_1\doteq g_1$ and $f_2\doteq g_2$ by four terms
    \begin{align*}
      \Phi_1 &:\inter{f_1}{\rho}\to\inter{g_1}{\rho} & \Phi_2 &:\inter{f_2}{\rho}\to\inter{g_2}{\rho}\\
      \Psi_1 &:\inter{g_1}{\rho}\to\inter{f_1}{\rho} & \Psi_2 &:\inter{g_2}{\rho}\to\inter{f_2}{\rho},
    \end{align*}
    we interpret $f_1 + f_2\doteq g_1 + g_2$ by using the terms
    \begin{align*}
      \Phi &= \lambda x. \case{x}{x_1}{\inl{(\Phi_1 x_1})}{x_2}{\inr{(\Phi_2 x_2)}}\\
      \Psi &= \lambda y. \case{y}{y_1}{\inl{(\Psi_1 y_1)}}{y_2}{\inr{(\Psi_2 y_2)}}.
    \end{align*}
    The fact that $\Phi$ and $\Psi$ are mutually inverse w.r.t. $=_{\beta\eta}$ is proved by using the $\eta$-axiom for sums twice. Once we use it with
    \begin{align*}
      M := &y\\
      N := &\casetwo{\big(\case{z}{y_1}{\inl{(\Psi_1 y_1)}}{y_2}{\inr{(\Psi_2 y_2)}}\big)}{x_1}{\inl{(\Phi_1 x_1)}}{x_2}{\inr{(\Phi_2 x_2)}},
    \end{align*}
    and the second time with $M:=y$, $N:=z$.
  \end{itemize}

  It remains to interpret the rest of the axioms of HSI*, those that involve $t_z$-terms.

  \begin{itemize}
  \item $t_1\doteq 1$ is interpreted as the isomorphism $\mathbf{1}\cong\inter{1}{\rho}$ i.e. $\mathbf{1}\cong \mathbf{1}$ using the lambda term $\lambda x.x$ in both directions;
  \item $t_{x_i}\doteq x_i$ is interpreted as $\mathbf{k}\cong \mathbf{k}$, for $k=\eval(x_i,\rho)$, by the lambda term $\lambda x.x$ in both directions;
  \item $t_{z u} \doteq t_z t_u$ is interpreted as $\mathbf{k}\cong \mathbf{k_1}\times \mathbf{k_2}$, for $k=\eval(t_{z u},\rho)$, $k_1=\eval(t_{z},\rho)$, and $k_2=\eval(t_{u},\rho)$, by the lambda term of Lemma~\ref{lemma};
  \item $t_{z + u} \doteq t_z + t_u$ is interpreted as $\mathbf{k}\cong \mathbf{k_1}+\mathbf{k_2}$, for $k=\eval(t_{z u},\rho)$, $k_1=\eval(t_{z},\rho)$, and $k_2=\eval(t_{u},\rho)$, by the lambda term of Lemma~\ref{lemma};
  \item $t_z \doteq t_u$ is interpreted as $\mathbf{k}\cong \mathbf{k}$, for $k=\eval(z,\rho)=\eval(u,\rho)$, by the lambda term $\lambda x.x$ in both directions.
  \end{itemize}

\end{proof}

\begin{lemma}\label{lemma} Let $p$$\in\Estar$ be an ordinary polynomial with positive coefficients (possibly containing $t_z$-terms), $k\in \Nplus$, and $\rho$ be an interpretation such that $k = \eval(p,\rho)$. Then, $\mathbf{k}\cong \inter{p}{\rho}$.
\end{lemma}
\begin{proof}
  We do induction on $p$.
  \begin{itemize}
  \item When $p= 1$, we have $\eval(1,\rho) = 1 = k$, so $\mathbf{k} = \mathbf{1} \cong \mathbf{1} = \inter{1}{\rho}$ is established by using $\lambda x.x$ in both directions.
  \item When $p=x_i$, we have $\eval(x_i,\rho) = \rho(x_i) = k$, so $\mathbf{k} \cong \inter{x_i}{\rho}$ by $\lambda x.x$ in both directions.
  \item When $p=t_z$, we have $\eval(t_z,\rho) = k$, so $\mathbf{k}\cong\inter{t_z}{\rho}$ by $\lambda x.x$ in both directions.
  \item When $p= p_1 + p_2$, we have 
    \begin{multline*}
      k=\eval(p,\rho)=\eval(p_1+p_2,\rho)=\\=\eval(p_1,\rho)+\eval(p_2,\rho)=k_1+k_2
    \end{multline*}
    for some $k_1,k_2\in\Nplus$. By applying the induction hypothesis twice, we obtain $\mathbf{k_1}\cong\inter{p_1}{\rho}$ and $\mathbf{k_2}\cong\inter{p_2}{\rho}$, therefore,
    \[
    \mathbf{k} \cong \mathbf{k_1} + \mathbf{k_2} \cong \inter{p_1}{\rho} + \inter{p_2}{\rho} = \inter{p}{\rho}
    \]
    using the obvious isomorphism $\mathbf{k}\cong \mathbf{k_1}+\mathbf{k_2}$, that holds when $k=k_1+k_2$ and the lambda terms interpreting congurence for ``+'' from the proof of the previous theorem.
  \item When $p = p_1 p_2$, we have 
    \begin{multline*}
      k=\eval(p,\rho)=\eval(p_1 p_2,\rho)=\\=\eval(p_1,\rho) \eval(p_2,\rho)=k_1 k_2
    \end{multline*}
    for some $k_1,k_2\in\Nplus$. By applying the induction hypothesis twice, we obtain $\mathbf{k_1}\cong\inter{p_1}{\rho}$ and $\mathbf{k_2}\cong\inter{p_2}{\rho}$, therefore,
    \[
    \mathbf{k} \cong \mathbf{k_1} \times \mathbf{k_2} \cong \inter{p_1}{\rho} \times \inter{p_2}{\rho} = \inter{p}{\rho}
    \]
    using the obvious isomorphism $\mathbf{k}\cong \mathbf{k_1}+\mathbf{k_2}$ that holds when $k=k_1 k_2$, and the lambda terms interpreting congurence for ``$\times$'' from the proof of the previous theorem.
  \end{itemize}
\end{proof}

\begin{corollary} Given two types $f, g\in\Ecal$, one can decide whether $\inter{f}{\rho} \cong\inter{g}{\rho}$ or not, and this holds whenever $\rho$ interprets variable by types of form $\mathbf{k}$.
\end{corollary}

\section{Conclusion}\label{sec:conclusion}


We showed that existing results from Mathematical Logic allow us to conclude that type isomorphism over $\Ecal$ is recursively axiomatizable, and that a subclass $\mathcal{L}$ of types can be isolated for which type isomorphism is even finitely axiomatizable by the well known High-School Identities and decidable. Our Theorem~\ref{thm:hsistar:sound} allows us to conclude decidability for the whole of $\Ecal$ when base types are finite sets.

These results also apply to questions of cardinality of sets in Constructive Mathematics, and to isomorphism of objects in the corresponding category. However, further work is needed to understand fully their implications in practice.

\subsection{Open Problems}

\subsubsection{Extensions and Practical Importance of the  Levitz Class}

We saw that the class $\mathcal{L}$ of Gurevič is a generalization of the classes $\mathcal{R}(S)$ and $\mathcal{L}(S)$ of Henson and Rubel, which are in turn generalizations of Levitz's class $\mathcal{S}$. We also saw in Example~\ref{example2} that there are two HSI-equal types, one of which is in $\mathcal{L}$ while the other is not. 

Therefore, it does not seem unlikely that the class $\mathcal{L}$ can be further extended. For example, cannot we allow the bases of exponentiation to contain variables in their bases of exponentiation up to a fixed (but arbitrary) height $n$? Would not such a theory also be finitely axiomatizable by HSI?

Another interesting thing to investigate would be what the practical interest of these subclasses is. For example, how many programs of a standard library for functional programming or theorem proving would fall outside (extensions of) $\mathcal{L}$?

\subsubsection{Simpler Completeness Proof for HSI*}
Wilkie's proof of completeness of HSI* relies on two components. 

In Theorem~2.8 of~\cite{wilkie00}, it is shown that each polynomial with exponentiation $f$ can be proved to be equal in HSI* to a positive polynomial with positive coefficients, but with extra variables, some of which are instantiated with witnessing terms $\tau_i$ of the form $p_i^{q_i}$. The proof of this theorem proceeds by induction on the construction of $f$, the difficult case being when $f= f_1^{f_2}$; here, the induction hypothesis is used together with the fact that each positive polynomial can be factored as a monomial and irreducible polynomial. In fact, an enumeration of all possible pairs $\langle\text{ irreducible },\text{ monomial }\rangle$ with the right properties is used, and then, when constructing the representative for $f$ one need only look up in the enumeration. A large number of extra variables will generally be added to the representing ``polynomial''.

The second component of Wilkie's proof uses Differential Algebra to show that the representation from Theorem~2.8 is unique.

If we could obtain a simpler version of Wilkie's proof of Theorem~2.8, that avoids an ad hoc enumeration, it would be easier to interpret Corollary~\ref{cor:hsistar:complete} as a program that given a concrete proof of $\typeof{f}\cong\typeof{g}$ builds a concrete derivation of $\HSIstar\vdash f\doteq g$.

\subsubsection{Efficient Decision Procedures}
The decision procedure of Macintyre \cite{macintyre81}, although primitive recursive, is of exponential complexity. It proceeds by computing, for a given equality one wants to test, an upper bound for the values of the variables, and then does a brute-force search up to that upper bound.

On the other hand, we know that equality between ordinary polynomials can be decided in $O(n \log(n))$ \cite{schwartz80}, and in the context of type isomorphisms decision algorithms with similar complexity have been given by Considine, Gil and Zibin \cite{considine00,gil05,gil07}.

Is to possible to obtain a decision algorithm for type isomorphism over $\Ecal$ with less-than-exponential complexity? Is it at least possible to do so for some subclass of terms like $\mathcal{L}$?




\subsection{Other Related Work}

Bruce, Di Cosmo and Longo \cite{bruce92} contains a more modern presentation of Soloviev's result \cite{soloviev83}.

Do\v{s}en and Petri\'{c} \cite{dosen97} show type isomorphism in the context of symmetric monoidal closed categories is finitely axiomatizable and decidable.

Di Cosmo and Dufour \cite{cosmo05} show that if one considers the structure $\Nplus\cup\{0\}$, decidability and non-finite-axiomatizability of arithmetical equality still hold. It is known that that there are true arithmetical equalities concerning $0$ that do not hold as type isomorphisms, an observation attributed to Alex Simpson in \cite{fiore06}. We have not treated the empty type here. This is motivated by our own interests in using type isomorphism over $\Ecal$: in intuitionistic logic it is not necessary to have a dedicated absurdity proposition, since as soon as we have basic arithmetic, the \emph{ex falso quodlibet} rule is derivable by induction.



\section*{Acknowledgments}
I would like to thank Olivier Danvy for being an excellent host at Aarhus University, where this project finally took off. I would also like to thank Alex Simpson for noticing an error at a workshop talk that I gave on the subject, and Gurgen Asatryan for sending me a copy of a hard-to-find paper.

This work is supported by a fellowship from Kurt Gödel Society, Vienna. Last changes were made while the author was working under the frame of the ProofCert ERC grant.



\bibliographystyle{IEEEtran}
\bibliography{sumaxioms}

\end{document}